\documentclass[11pt]{amsart}
\usepackage{amsmath}
\usepackage{amsfonts}
\usepackage{color}      
\usepackage{epsfig}
\usepackage{mathtools}
\usepackage{amsthm}
\usepackage{bigints}

\allowdisplaybreaks

\usepackage{hyperref}

\begin{document}

\newtheorem{theorem}{Theorem}
\newtheorem*{theorem*}{Theorem}
\newtheorem{lemma}{Lemma}
\newtheorem{remark}{Remark}
\newtheorem{proposition}{Proposition}
\newtheorem{definition}{Definition}
\newtheorem*{proposition*}{Proposition}

\newcommand{\tr}{\operatorname{tr}}
\newcommand{\Hess}{\operatorname{Hess}}

\newcommand\norm[1]{\left\lVert#1\right\rVert}

\newcommand{\comment}[1]{}

\numberwithin{equation}{section}

\title[Hyperbolic Einstein Constraint Equations and the Penrose Inequality]{Asymptotically Hyperbolic Einstein Constraint Equations with Apparent Horizon Boundary and the Penrose Inequality for Perturbations of Schwarzschild-AdS}

\author{Marcus Khuri}
\address{Department of Mathematics, Stony Brook University, 
Stony Brook, NY 11794, USA}
\email{khuri@math.sunysb.edu}

\author{Jaros\l aw Kopi\'nski}
\address{Center for Theoretical Physics, Polish Academy of Sciences,
Warsaw, Poland}
\email{jkopinski@cft.edu.pl}
\thanks{M. Khuri acknowledges the support of NSF Grant DMS-2104229, and Simons Foundation Fellowship 681443. J. Kopi\'nski acknowledges support from the Norwegian Financial Mechanism 2014-2021 (UMO-2019/34/H/ST1/00636), and a Ko\'sciuszko Foundation Grant.}

\begin{abstract}
We prove the existence of asymptotically hyperbolic solutions to the vacuum Einstein constraint equations with a marginally outer trapped boundary of positive mean curvature, using the constant mean curvature conformal method. As an application of this result, we verify the Penrose inequality for certain perturbations of Schwarzschild Anti-de Sitter black hole initial data.
\end{abstract}

\maketitle

\section{Introduction}
\label{sec1}

An initial data set for the vacuum Einstein equations with cosmological constant $\Lambda<0$ consists of a triple $(M, \hat{g}, \hat{k})$, where $M$ is a 3-dimensional Riemannian manifold with metric $\hat{g}$ and $\hat{k}$ is a symmetric 2-tensor that satisfy
\begin{equation}\label{hamcon}
R_{\hat{g}} + ( tr_{\hat{g}}\hat{k} )^2  - |\hat{k}|^2_{\hat{g}} + 2 |\Lambda|= 0, \quad\quad\quad
\mathrm{div}_{\hat{g}} \hat{k} - d ( \mathrm{tr}_{\hat{g}} \hat{k} )= 0 ,
\end{equation}
in which $R_{\hat{g}}$ denotes scalar curvature. These are the constraint equations that represent vanishing energy and momentum densities of the matter fields, and arise from traces of the Gauss-Codazzi system when the data are viewed as embedded in spacetime. In order to study isolated gravitating systems, it will be assumed that the data are asymptotically hyperbolic in the sense that the extrinsic curvature $\hat{k}$ falls-off, and the metric $\hat{g}$ asymptotes to the hyperbolic metric in an asymptotic end. Constructions of asymptotically hyperbolic initial data sets have previously been studied by several authors including Andersson-Chru\'sciel \cite{andersson_chrusciel}, Gicquaud \cite{gicquaud, sakovich_ncmc}, and Sakovich \cite{sakovich}. See also the recent related work of Allen-Lee-Maxwell \cite{maxwell_spaces}.

In the current paper, we will study the problem of constructing asymptotically hyperbolic initial data with apparent horizon, or rather marginally outer trapped surface (MOTS), boundary conditions in which the mean curvature is positive. Recall that a closed 2-sided surface $\Sigma\subset M$ is future trapped if its null expansion satisfies $\theta_+ =\hat{h}+\mathrm{tr}_{\Sigma}\hat{k}<0$, where $\hat{h}$ denotes mean curvature with respect to the unit normal pointing towards the asymptotic end. From the spacetime perspective, the null expansion represents the mean curvature in null directions, and thus measures the rate of change of area of shells of light emanating from the surface in the outward future direction. The trapped condition is interpreted as signifying a strong gravitational field. Moreover, MOTS are defined by the equation $\theta_+ =0$ and arise as the boundary of trapped regions \cite{AnderssonMetzger,Eichmair}. These surfaces may be interpreted as quasi-local versions of the event horizon within initial data, since they depend only on the local geometry of a slice whereas the event horizon requires global knowledge of the spacetime. Previous work centered on initial data construction with MOTS boundary, in the asymptotically flat context, has been carried out by Bowen-York \cite{bowen_york}, Dain \cite{dain}, Holst-Meier \cite{holst}, Maxwell \cite{maxwell,maxwell2}, and Thornburg \cite{thornburg} among others.

A standard approach to solving the Einstein constraint equations is the conformal method of Lichnerowicz \cite{lichnerowicz}, Choquet-Bruhat and York \cite{chby}, where (in vacuum) the conformal metric, the divergence-free and traceless part of the `conformal extrinsic curvature', and the mean curvature of the slice are treated as free data. This reduces the constraint equations to a coupled elliptic system, where the unknowns are the conformal factor and the remaining portion of the conformal extrinsic curvature. Moreover, the mean curvature serves as a coupling function in the reduced system, so that in the CMC (constant mean curvature) case the Hamiltonian and momentum constraints decouple and can be solved independently. Initial data containing black holes fits well into this approach. The assumption that the boundary is a marginally outer trapped surface can be encoded in the form of a Robin-type boundary condition for the conformal factor. An extended discussion of results and approaches to the conformal method can be found in reviews by Bartnik-Isenberg \cite{BartnikIsenberg}, and Carlotto \cite{carlotto}.

As an application of our existence result, concerning asymptotically hyperbolic initial data admitting a MOTS boundary with positive mean curvature, we will confirm the Penrose inequality for certain perturbations
of Schwarzschild-AdS. The Penrose inequality is a precise lower bound for the total mass of an initial data set
in terms of the (appropriately defined) surface area of black holes contained within it. Although it was originally
conjectured in the asymptotically flat setting \cite{penr_ns}, where it was proven in time-symmetry by Bray \cite{bray} and Huisken-Ilmanen \cite{huisken_ilmanen}, the Penrose inequality has been proposed with two different
versions in the asymptotically hyperbolic context \cite{braychrusciel,wang}. Namely, one form of the inequality is suited for asymptotically hyperboloidal slices of asymptotically flat spacetimes, while the other is tailored for asymptotically
totally geodesic slices of asymptotically AdS spacetimes; it is the latter that will be studied here. Relatively little is known about the hyperbolic versions of the Penrose inequality outside of spherical symmetry \cite{husain}, the graphical case 
\cite{dahl_sakovich,lima_girao}, and time-symmetric perturbations of Schwarzschild-AdS \cite{ambrosio}. See also the results of \cite[Theorem 1.3]{ahk} and \cite{neves} concerning asymptotically locally hyperbolic initial data. Moreover, it is known that the naive approach using inverse mean curvature flow does not succeed \cite{neves2}, however when coupled with a Jang-type equation the desired inequality follows \cite{mk_cha,chakhurisakovich} assuming existence for a coupled system of equations.

This paper is organized as follows. In the next section the main results will be stated, while in Section \ref{prel} we make the appropriate definitions, set notation, and review various details of the conformal method. Section \ref{sec4} is dedicated to a mean curvature estimate, and in Section \ref{secex} we construct appropriate barriers for application of the method of sub/super solutions to solve the Lichnerowicz equation. The purpose of Section \ref{secpenr} is to record properties of initial data with a MOTS boundary that are conformal to Schwarzschild-AdS, in preparation for studying the Penrose inequality in a perturbative regime. Lastly, in Section \ref{secpenr2} we establish the spacetime Penrose inequality for a class of conformal perturbations of Schwarzschild-AdS data.


\section{Statement of Results}
\label{sec2}

The main results of this work consist of two theorems. The first is a statement concerning the existence of vacuum asymptotically hyperbolic maximal initial data having a MOTS boundary of positive mean curvature. This is established with the conformal method, and comes with an assumption that the boundary mean curvature with respect to the seed metric $g$ is nonnegative and bounded above by a constant $\mathcal{C}$ depending on the boundary normal injectivity radius $T_0$, the sectional curvature $\mathrm{K}$ of $g$, and the cosmological constant $\Lambda$. Namely
\begin{equation} \label{h_bound_introduction}
\max \displaylimits_{\partial M} h < \mathcal{C} \left(T_0,  \min \displaylimits_{M_{T_0}} \mathrm{K}, \max \displaylimits_{M_{T_0}} \mathrm{K}^+, \Lambda \right),
\end{equation}
where $\mathrm{K}^+$ is the positive part of $\mathrm{K}$, and $M_{T_0}$ denotes the set of points whose distance from the boundary is not larger than $T_0$. While the precise form of this constant is given in Proposition \ref{existence}, here we note that $\mathcal{C}$ reduces to $\min\{ 4 ,4  |\Lambda| \}$ for boundaries with an infinite injectivity radius in manifolds of nonpositive curvature. Another hypothesis of the result is that the seed Riemannian manifold is asymptotically hyperbolic of constant scalar curvature, and with a boundary of positive mean curvature.
However, such seeds are readily available as they may be obtained by conformal deformation from essentially any asymptotically hyperbolic manifold \cite[Theorem 1.1]{gicquaud}. This hypothesis may be compared with a similar one in the asymptotically flat setting, where in \cite[Corollary 1]{maxwell} seed data are produced by conformal change to zero scalar curvature and small positive boundary mean curvature\footnote{In the statement of \cite[Corollary 1]{maxwell} the mean curvature is said to be negative, however the author there is using an alternative definition of second fundamental form, resulting in a change of sign.}. The weighted H\"older spaces used in the following result are presented in the next section. 

\begin{theorem} \label{mthi}
Let $(M,g)$ be a 3-dimensional $C^{2,\beta}_{\tau}$-asymptotically hyperbolic manifold with $0<\beta<1$, $1 < \tau < 3$, scalar curvature $R_g = - 2 |\Lambda|$, and boundary of nonnegative mean curvature $h$ that satisfies \eqref{h_bound_introduction}. Assume that $k \in C^{1, \beta}_{\tau}(M)$ is a divergence and trace-free 2-tensor with $0 \leq k_{nn} \leq  h$ on $\partial M$, where $n$ denotes the unit normal pointing towards the asymptotic end. Then the boundary value problem
\begin{equation} \label{mthbvp1}
\begin{split}
\Delta_g \phi &  = \frac{|\Lambda|}{4} \phi \left( \phi^4 -1 \right) - \frac{1}{8 } |k|^2_{g} \phi^{-7} \quad \mathrm{on} \quad M,
\\ \partial_n \phi &  = - \frac{1}{4}h \phi+ \frac{1}{4 }k_{nn}\phi^{-3}  \quad \mathrm{on} \quad \partial M,
\end{split}
\end{equation}
admits a positive solution with $\phi-1 \in C^{2,\beta}_{\tau}(M)$. In particular, the metric
$\hat{g} = \phi^4 g$ and tensor $\hat{k} = \phi^{-2} k$ satisfy the vacuum constraints \eqref{hamcon}, and $(M, \hat{g}, \hat{k})$ forms an asymptotically hyperbolic maximal initial data set of class $\left(\beta, \tau \right)$ with a MOTS boundary of nonnegative mean curvature $\hat{h}=\phi^{-6}k_{nn}$.
\end{theorem}

\begin{remark}
Let $\kappa$ be a smooth section of $T^* M$ restricted to $\partial M$. Then Theorem \ref{akfhoiuqoih} below shows how to construct a transverse-traceless (divergence and trace-free) tensor $k \in C^{1, \beta}_{\tau}(M)$ such that $k(n,\cdot)=\kappa$ on $\partial M$. Hence, for any function $\kappa(n)\in C^{\infty}(\partial M)$ satisfying $0\leq \kappa(n)\leq h$, there exists $k$ satisfying the relevant hypotheses of Theorem \ref{mthi}.
\end{remark}

This theorem is complementary to \cite[Theorem 1.2]{gicquaud}, where the resulting initial data has a MOTS of nonpositive mean curvature because of the assumption there that $k_{nn} \leq 0$ on $\partial M$. Theorem \ref{mthi} admits MOTS of positive $\hat{h}$, at the cost of placing an upper bound \eqref{h_bound_introduction} on the seed mean curvature in terms of the  geometry near the boundary. The positive sign of $\hat{h}$ is relevant from the perspective of proving the Penrose inequality,
where it is used in the perturbative setting to show that the boundary is outerminimizing; for a precise statement regarding this issue see Lemma \ref{aeoiroiqh} below. The proof of Theorem \ref{mthi} is motivated by a similar result for the asymptotically flat case treated by Maxwell \cite[Theorem 1]{maxwell}, where the hypothesis of a positive Yamabe invariant is used and may be viewed as related to the mean curvature upper bound. Indeed, when combined with a seed metric of zero scalar curvature, positivity of the Yamabe invariant \cite[page 563]{maxwell} yields an integral upper bound for the mean curvature in terms of global quantities.

The second main result is a verification of the Penrose inequality for perturbations of Schwarzschild-AdS initial data. A similar study has been carried out by the second author and Tafel in \cite{kt1, kt2} for asymptotically flat and axisymmetric initial data, where the Penrose inequality with angular momentum was confirmed up to second order of expansion with respect to a scale determined by the conformal extrinsic curvature. Moreover, the Penrose inequality for perturbations of Schwarzschild-AdS has been established in the time-symmetric case by Ambrozio \cite{ambrosio}.
In the current work, we will consider the case of non-time-symmetric perturbations, in particular those with a MOTS boundary.
We start with seed data $(g_S,k)$, where $g_S$ is the induced metric on the canonical slice of Schwarzschild-AdS and $k$ is a transverse-traceless tensor with respect to $g_S$. Recall that in Schwarzschild coordinates the seed metric of mass $m>0$ is defined on $M=[\rho_h,\infty)\times S^2$ and is given by
\begin{equation}
g_{S} =  \frac{d\rho^2}{1-\frac{2m}{\rho} + \frac{|\Lambda|}{3} \rho^2} + \rho^2 g_{S^2},
\end{equation}
where $\rho_h>0$ is the real zero of the static potential, and $g_{S^2}$ denotes the unit round metric on the 2-sphere. In order to obtain perturbations of the canonical slice, we take solutions to the constraints via the conformal method such that the conformal factor and seed extrinsic curvature satisfy
\begin{equation} \label{fact_exp145}
\phi = 1 + \sum \displaylimits_{j = 1}^{\infty}  \phi_j \epsilon^j,
\quad\quad\quad k = \epsilon k',
\end{equation}
for a perturbation parameter $\epsilon>0$, functions $\phi_{j} \in C^{2,\beta}_{\tau}(M)$, and a fixed transverse-traceless tensor $k'\in C^{1,\beta}_{\tau}(M)$. Further discussion on the construction of the perturbed data is given in Section \ref{6.3}. In the following result, $dV_S$ and $H^1$ will represent the volume form and Sobolev space of square integrable derivatives with respect to $g_S$.

\begin{theorem} \label{mthi2}
Let $(M, \hat{g}=\phi^4 g_S, \hat{k}=\phi^{-2}k)$ be a 3-dimensional, vacuum, maximal, asymptotically hyperbolic, conformally perturbed Schwarzschild-AdS initial data set as in \eqref{fact_exp145} of class $\left(\beta, \tau \right)$ with $0<\beta<1$, $\tau>3/2$, and having MOTS inner boundary $\partial M$. If there exists a constant $C$ and domain $M'\supset \partial M$ such that
\begin{equation}
\int \displaylimits_M  |k|^2_{g_S} \sqrt{1-\frac{2m}{\rho} + \frac{|\Lambda|}{3} \rho^2 }  d V_S > C  ||k_{nn}||^2_{H^1 \left( M' \right)},
\end{equation}
then the Penrose inequality
\begin{equation}
\mathfrak{m} \geq  \sqrt{\frac{A_h}{16 \pi}} \left(1+ \frac{ |\Lambda| }{12 \pi} A_h  \right)
\end{equation}
holds for all $\epsilon$ sufficiently small, where $\mathfrak{m}$ denotes total mass and $A_h$ is the area of the MOTS boundary.
\end{theorem}


\section{Preliminaries} \label{prel}

\subsection{Asymptotically hyperbolic initial data}
Let $\Lambda<0$ and consider the reference hyperbolic space $(\mathbb{H}^3_{\Lambda}, b)$ of curvature $\Lambda/3$, where the metric is given in scaled geodesic polar coordinates by
\begin{equation} \label{refmet}
b=  \frac{3}{|\Lambda|}\left( d r^2 + \sinh^2 r g_{S^2} \right),
\end{equation}
and $g_{S^2}$ is the unit round metric on the 2-sphere. We will work with definitions of weighted H\"older spaces as presented in \cite{lan jang}; see \cite{maxwell_spaces} for recent developments on the use of weighted spaces in the asymptotically hyperbolic context. Let $B\subset\mathbb{H}^3_{\Lambda}$ be a ball, $l\in\mathbb{N}$, $\beta\in[0,1]$, $\tau\in\mathbb{R}$ and define $C^{l,\beta}_{\tau} \left( \mathbb{H}^3_{\Lambda} \setminus B \right)$ to be the collection of functions that are locally H\"older smooth of order $(l,\beta)$ on the complement domain and satisfy
\begin{equation}
	| f |_{C^{l,\beta}_{\tau}  \left( \mathbb{H}^3_{\Lambda} \setminus B \right) } := \sum_{i =0}^l \sup_{x\in \mathbb{H}^3_{\Lambda} \setminus B} e^{\tau r } |\mathring{\nabla}^i f(x)|_b +\sup_{x\in \mathbb{H}^3_{\Lambda} \setminus B} e^{\tau r } [\mathring{\nabla}^l f]_{\beta; B_1(x)} <\infty,
\end{equation}
where $\mathring{\nabla}$ denotes covariant differentiation with respect to $b$ and
\begin{equation}
	[\mathring{\nabla}^l f]_{\beta; B_1(x)} =\sup_{1\le i_1, \dots, i_l \le 3} \left( \sup_{y\neq z\in B_1(x) } \frac{|e_{i_1} \cdots e_{i_l} (f)(y) - e_{i_1} \cdots e_{i_l} (f) (z) |}{d_b(y,z)^\beta} \right),
\end{equation}
where $B_1(x)$ is the unit geodesic ball centered at $x$ intersected with $\mathbb{H}^3_{\Lambda} \setminus B$, the $b$-distance is labelled $d_b$, and $\{e_i \}$ is the orthonormal frame for $b$ adapted to the geodesic polar coordinates used in \eqref{refmet}. This definition may be extended to tensors of arbitrary type by requiring that their components (in the orthonormal frame) lie inside $C^{l,\beta}_{\tau} \left( \mathbb{H}^3_{\Lambda} \setminus B \right)$. Furthermore, if $\mathcal{K}\subset M$ is a compact set and $\Psi: M \setminus \mathcal{K} \to \mathbb{H}^3_{\Lambda} \setminus B$ is a diffeomorphism, then
the weighted H\"older norms $|f|_{C^{l,\beta}_{\tau} \left( M \right)}$ may be defined as the sum of the weighted norm $|\Psi_* f|_{C^{l,\beta}_{\tau} \left( \mathbb{H}^3_{\Lambda} \setminus B \right)}$ and the typical $C^{l, \beta}$ norm over $\mathcal{K}$. The completion of the spaces $C^{l,\beta}_{c} \left( M \right)$ of compact support, with respect to the weighted H\"older norm, is then denoted $C^{l,\beta}_{\tau} \left( M \right)$.

\begin{definition} \label{def1}
We say that a 3-dimensional Riemannian manifold $(M,g)$ is $C^{l,\beta}_{\tau}$-asymptotically hyperbolic for $l \in \mathbb{N}$, $\beta\in[0,1]$, and $\tau >0$, if there exists a compact set $\mathcal{K} \subset M $ and a diffeomorphism
$\Psi: M \setminus \mathcal{K} \to \mathbb{H}^3_{\Lambda} \setminus B$ such that
\begin{equation}
\Psi_*  g - b \in C^{l,\beta}_{\tau} \left(  \mathbb{H}^3_{\Lambda} \setminus B \right).
\end{equation}
\end{definition}


\begin{definition}
A triple $(M, g, k)$ is an asymptotically hyperbolic initial data set of class $(\beta, \tau)$ with  $\beta\in[0,1]$ and $\tau >0$ if
\begin{itemize}
\item $(M, g)$  is a $C^{2, \beta}_{\tau}$-asymptotically hyperbolic manifold,
\item $k$ is a symmetric 2-tensor and $k \in C^{1, \beta}_{\tau} (M)$.
\end{itemize}
\end{definition}

\subsection{Conformal method for the Einstein constraint equations}
In this well-studied approach to the vacuum constraints, a triple $(M,g, \varsigma)$ consisting of Riemannian manifold $(M,g)$ and a scalar function (prescribed mean curvature) $\varsigma$ is given. The constraints are then solved by searching for a scalar function $\phi>0$ (the conformal factor) and a one-form $X$ that satisfy
\begin{align}
\Delta_{g} \phi  - \frac{1}{8} R_g \phi =& \frac{1}{4}\left(\frac{1}{3} \varsigma^2 + |\Lambda| \right) \phi^5  - \frac{1}{8} |k|^2_{g} \phi^{-7}, \label{hamcon1}\\
 \mathrm{div}_g \left( \mathcal{D}_g X \right) = & \frac{2}{3} \phi^6 d \varsigma, \label{momcon1}
\end{align}
where $R_g$ is scalar curvature, $\Delta_g$ is the Laplace operator of $g$, the conformal extrinsic curvature is given by $k = \mathcal{D}_g X$, 
and $\mathcal{D}_g$ represents the trace-free Lie derivative
\begin{equation}
\mathcal{D}_g X  := \mathcal{L}_{X} g - \frac{2}{3} \left( \mathrm{div}_g X \right)g.
\end{equation}
The semilinear elliptic equation (\ref{hamcon1}) is called the Lichnerowicz equation \cite{lichnerowicz}. Furthermore,
the desired solution initial data set $(M, \hat{g}, \hat{k})$ for the vacuum constraint equations may then be obtained from these quantities as follows
\begin{equation}
\hat{g} = \phi^4 g, \quad\quad\quad
\hat{k} = \phi^{-2} k+ \frac{1}{3} \varsigma \phi^4 g.
\end{equation}
If $\varsigma = \mathrm{const}$ (CMC data) then the conformally formulated constraints decouple, and the solution can be obtained by firstly solving (\ref{momcon1}) for $X$ and then solving (\ref{hamcon1}) for $\phi$. Alternatively, if one is given a transverse-traceless tensor $k$ on $(M,g)$ to begin with, then it is only necessary to solve the Lichnerowicz equation, and we refer to $(M,g,k)$ as seed data.

\subsection{Marginally outer trapped surface boundary condition}
Suppose that a solution to the constraints $(M,\hat{g},\hat{k})$ is given with an asymptotic end, as described above.
Then the future null expansion of an inner boundary $\partial M$ takes the form
\begin{equation} \label{thetaplus}
\theta_+ (\partial M) = \varsigma - \hat{k}_{\hat{n} \hat{n}} + \hat{h},
\end{equation}
where $\hat{h}= \hat{\mathrm{div}}_{\partial M} \hat{n}$ is mean curvature and $\hat{n}$ is the boundary unit normal with respect to $\hat{g}$ that is pointing towards the asymptotic end. In terms of conformal quantities we have
\begin{equation}
\theta_+ (\partial M) = \frac{2}{3}  \varsigma  -  k_{nn}\phi^{-6} +  \left(4 \partial_n \phi + h \phi \right)\phi^{-3},
\end{equation}
where $h$ and $n$ are the boundary mean curvature and unit normal with respect to $g$. The boundary is a MOTS if
$\theta_+ (\partial M)=0$, which reduces to a nonlinear Robin-type boundary condition for $\phi$ that may be expressed as
\begin{equation}
\partial_n \phi + \frac{h}{4} \phi = \frac{1}{4} k_{nn}  \phi^{-3}-\frac{\varsigma}{6} \phi^3 \quad \mathrm{on} \quad \partial M.
\end{equation}
The physical interpretation of the MOTS condition is that light rays emanating from such a surface to the future are not diverging. Under favorable conditions, this indicates the existence of a black hole region within the spacetime evolved from the initial data.

\subsection{Maximal initial data with a MOTS}

In summary, if the vacuum initial data set $(M, \hat{g}, \hat{k})$ has vanishing mean curvature ($\varsigma = 0$) and an inner boundary $\partial M$ that is a MOTS, then the conformally formulated constraint equations reduce to the boundary value problem
\begin{equation} \label{constraintbvp}
\begin{split}
\Delta_g \phi   - \frac{1}{8} R_g \phi & = \frac{1}{4} |\Lambda|  \phi^5- \frac{1}{8} |k|^2_{g}\phi^{-7} \quad \mathrm{on} \quad M,
\\ \partial_n \phi + \frac{1}{4}h \phi &  =  \frac{1}{4}k_{nn}\phi^{-3}   \quad \mathrm{on} \quad \partial M,
\end{split}
\end{equation}
in addition to the vector equation for the extrinsic curvature
\begin{equation} \label{vecconcmc}
 \mathrm{div}_g \left( \mathcal{D}_g X \right) =0.
\end{equation}
In what follows we will study the existence of solutions to (\ref{constraintbvp}) and (\ref{vecconcmc}) in the asymptotically hyperbolic setting, under the CMC assumption. Note that the asymptotically hyperbolic condition implies that the only constant mean curvature possible is $\varsigma=0$. 

\section{Mean Curvature Estimate}
\label{sec4}

In this section we will prove a lower bound on the mean curvature of constant distance surfaces to the boundary of a Riemannian manifold, assuming that the boundary is of nonnegative mean curvature. This estimate will play an important role in the construction of barriers for the Lichnerowicz equation.

\begin{proposition} \label{meancurve}
Let $(M,g)$ be a 3-dimensional complete Riemannian manifold with boundary $\partial M$, and set $M_T = \{x \in M\mid t(x) \leq T \}$ where $t(x)= \mathrm{dist} \left(x, \partial M \right)$. Assume that $T$ is smaller than the normal injectivity radius of $\partial M$, and that its mean curvature is non-negative. Then
\begin{equation} \label{hminineq}
\max \displaylimits_{M_T} |h^-_t| \leq  T \cosh^2 ( \sqrt{|\mathrm{K}_0| } T )   \sum \displaylimits_{i=1,2}   \max \displaylimits_{M_T} \mathrm{K}^+ (\partial_t, e_i  ) ,
\end{equation}
where $h_t^-$ is the negative part of the mean curvature for the $t$-level sets, $\mathrm{K}^+$ is the positive part of the sectional curvature, $\mathrm{K}_0$ is a negative lower bound for sectional curvature in $M_T$, and $\{ \partial_t, e_1, e_2 \}$ is an orthonormal frame.
\end{proposition}

\begin{proof}
Choose $T<T_0$, where $T_0$ is the normal injectivity radius of $\partial M$. The mean curvature $h_t$ of the $t$-level sets can be estimated from below in the following way. Consider a Jacobi field $J$ orthogonal to $\partial_t$ along a radial geodesic emanating from the boundary. Then from the Riccati equation we have
\begin{equation}
\partial_t \left(A_{JJ} \right)  = - \mathrm{K} \left(\partial_t, J\right) |J|^2_g + A^2_{JJ} 
\geq -  \max \displaylimits_{M_T} \mathrm{K}^+ \left(\partial_t, J \right)   \max \displaylimits_{M_T} |J|^2_g, 
\end{equation}
where $A$ is the second fundamental form of constant $t$-surfaces, and $\mathrm{K}(\partial_t,J)$ denotes sectional curvature. Hence
\begin{equation} \label{ajj}
A_{JJ}(t)   \geq -  t \left( \max \displaylimits_{M_T}   \mathrm{K}^+ \left(\partial_t, J \right)  \max \displaylimits_{M_T} |J|^2_g  \right)  + A_{JJ}(0),
\end{equation}
for $t\leq T$. Next, consider the Jacobi field Dirichlet problems
\begin{equation} \label{jaceq}
\nabla_t \nabla_t J_i + R(J_i, \partial_t) \partial_t =0,\quad\quad
J_i(0)=v_i,\quad\quad J_i(T)=w_i,\quad\quad i=1,2,
\end{equation}
where $\{v_1,v_2\}$ and $\{w_1,w_2\}$ are orthonormal eigenvectors of the second fundamental form $A$ at $t=0$ and $t=T$, respectively. Note that since there are no conjugate points for $t \leq T < T_0$, a unique solution exists to these boundary value problems. Using $J_i$, $i=1,2$ from (\ref{jaceq}) and taking a trace of (\ref{ajj}) yields
\begin{equation} \label{hest}
\max \displaylimits_{M_T} |h^-_t| \leq T  \sum \displaylimits_{i=1,2} \left(  \max \displaylimits_{M_T}   \mathrm{K}^+ \left(\partial_t, J_i  \right)     \max \displaylimits_{M_T}|J_i|^2_g \right),
\end{equation}
where the non-negativity of $h_0$ has been used. Observe that if $\mathrm{K}^+=0$, then the constant distance surfaces have nonnegative mean curvature ($h_t^-=0$). 

To estimate the size of the Jacobi field we will use the Rauch comparison theorem for hypersurfaces  (see \cite[Theorem 4.3]{warner}). Let $\mathrm{K}_0<0$ be a lower bound for the sectional curvature of $M_T$, and consider the hyperbolic space $(\mathbb{H}^3,g^{\mathrm{K}_0})$ of curvature $\mathrm{K}_0$ with metric
\begin{equation}
g^{\mathrm{K}_0} = d s^2 + \frac{1}{|\mathrm{K}_0|}  \sinh^2 \left( \sqrt{|\mathrm{K}_0| } s \right) g_{S^2}.
\end{equation}
The two eigenvalues of the second fundamental form for the $s= s_0$ coordinate sphere are given by
\begin{equation}
\lambda^{\mathrm{K}_0}_i ( s_0 )  = \sqrt{|\mathrm{K}_0|} \coth \left( \sqrt{|\mathrm{K}_0| }  s_0\right),\quad\quad i=1,2.
\end{equation}
Moreover, if $z(s)$ is a parallel transported unit vector field along and orthogonal to a radial geodesic in hyperbolic space, then
\begin{equation}
J^{\mathrm{K}_0} (s)  =  \frac{C}{\sqrt{|\mathrm{K}_0|}} \cosh \left( \sqrt{|\mathrm{K}_0| }  s\right) z(s) 
\end{equation}
is a Jacobi field for any constant $C$. We may choose $s_0 >0$ sufficiently small so that
\begin{equation}
\lambda_i (0) \leq \lambda^{\mathrm{K}_0}_i ( s_0 ) = \frac{1}{s_0} + O (1),
\end{equation}
where $\lambda_i (0)$ are the eigenvalues of the extrinsic curvature of $\partial M$, and we can arrange for $J^{\mathrm{K}_0}$ to have unit norm at $s=s_0$ by setting 
\begin{equation}
C = \frac{\sqrt{|\mathrm{K}_0|}}{  \cosh \left( \sqrt{|\mathrm{K}_0| }  s_0 \right)}
\quad\quad \Rightarrow \quad\quad
|J^{\mathrm{K}_0} \left( s_0 \right)|_{g^{\mathrm{K}_0}} = |J_i \left( 0\right)|_g =1.
\end{equation}
Therefore, by applying \cite[Theorem 4.3]{warner} it follows that
\begin{equation} \label{jacestim}
\begin{split}
& |J_i (t) |_g \leq  |J^{\mathrm{K}_0} (s_0 + t)|_{g^{\mathrm{K}_0}} = \frac{\cosh\left( \sqrt{|\mathrm{K}_0| } \left(  s_0 +t \right) \right)}{\cosh \left( \sqrt{|\mathrm{K}_0| } s_0 \right)} \quad \mathrm{for} \quad t \in [0,T].
\end{split}
\end{equation}
Since (\ref{jacestim}) is satisfied for all appropriately small $s_0$, we can take the limit $s_0 \to 0$ to obtain
\begin{equation}\label{aohoiyhgq}
|J_i (t) |_g \leq \cosh\left( \sqrt{|\mathrm{K}_0| } t \right) \quad \mathrm{for} \quad t \in [0,T], \quad i=1,2.
\end{equation}
Hence, the desired inequality (\ref{hminineq}) now follows from (\ref{hest}) and \eqref{aohoiyhgq}.
\end{proof}

\section{Constraint Equations With a MOTS Boundary Condition} \label{secex}

Let $(M,g)$ be a $C^{2,\beta}_{\tau}$-asymptotically hyperbolic manifold with a boundary $\partial M$, such that $0<\beta<1$ and $1<\tau<3$. We will use a barrier method following Gicquaud \cite[Proposition 2.1]{gicquaud} to solve the Lichnerowicz equation with MOTS boundary condition (\ref{constraintbvp}). The result of \cite{gicquaud} only gives rise to initial data with a MOTS boundary that has nonpositive mean curvature, while here we obtain initial data with a MOTS boundary of positive mean curvature. The construction of barriers in this latter case requires additional control of the geometry of the seed metric near the boundary, which is recorded below in Proposition \ref{existence}. It should be noted that a similar formulation of the barrier method was utilized by Maxwell \cite{maxwell}, but with the local control on the near-boundary geometry replaced with a restriction on the Yamabe invariant, in the context of asymptotically flat initial data with MOTS boundary condition. In addition, Sakovich \cite{sakovich} studied the Lichnerowicz equation with matter on asymptotically hyperbolic manifolds without inner boundary.

The proof of the main theorem will rely on the existence of a solution to a linear Robin boundary value problem on the asymptotically hyperbolic manifold. It may be interpreted as a model equation for the nonlinear Lichnerowicz equation with MOTS boundary condition, and will be used in the construction of the global supersolution for the latter. In order to solve this model equation we require a bound on the mean curvature of the inner boundary of the form
\begin{equation}
\max \displaylimits_{\partial M} h < \mathcal{C} \left( T_0, \mathrm{K}_0, \mathrm{K}_1,\Lambda\right),
\end{equation}
where $T_0$ is the injectivity radius of the boundary and $\mathrm{K}_0$ and $\mathrm{K}_1$ are sectional curvature lower and upper bounds in a neighborhood of $\partial M$. In the limit as  $|\Lambda| \to 0$ the function $\mathcal{C}$ approaches zero, whereas for boundaries with infinite injectivity radius in manifolds of nonpositive curvature it is equal to $\min\{4, 4 |\Lambda|\}$. Moreover, if there exists a foliation of positive mean curvature surfaces emanating from $\partial M$, the form of $\mathcal{C}$ simplifies in a significant way and no longer depends on $\mathrm{K}_0$ and $\mathrm{K}_1$.

\begin{proposition} \label{existence}
Let $(M,g)$ be a $C^{2,\beta}_{\tau}$-asymptotically hyperbolic manifold with $0<\beta<1$, and $\tau >0$. Assume that the mean curvature $h$ of the boundary $\partial M$ is nonnegative and bounded from above in the following way
\begin{equation} \label{existence_bound}
\max \displaylimits_{\partial M} h  < \mathcal{C} \left( T_0, \mathrm{K}_0, \mathrm{K}_1, \Lambda \right):=  \sup \displaylimits_{T \in \left[0,T_0 \right)} \left(  \frac{4 T\min\{1,|\Lambda|\}}{ 8 + T +2 T^2 \mathrm{K}_1  \cosh^{2} ( \sqrt{|\mathrm{K}_0| } T )  } \right) ,
\end{equation}
where $T_0$ is the boundary injectivity radius, $\mathrm{K}_0 <0 $ is a lower bound for sectional curvature in $M_T$, and $\mathrm{K}_1 : =  \max \displaylimits_{M_T} \mathrm{K}^+$ is the maximum of the positive part of the sectional curvature $\mathrm{K}$. Then the boundary value problem
\begin{equation}
\begin{split}
&\left( \Delta_g  - |\Lambda| \right)u = F \quad \mathrm{on} \quad  M, \\
& \left( \partial_n +  \frac{\alpha}{4}h \right)u = f \quad \mathrm{on} \quad \partial M,
\end{split}
\end{equation}
with $F \in C^{0, \beta}_{\delta} (M) $, $f \in C^{1, \beta} ( \partial M) $  and $\alpha \in [0,1]$ admits a solution $u$ in $C^{2,\beta}_{\delta} (M)$ for $1<\delta <3 $.
\end{proposition}

\begin{proof}
First note that the operator $P=\left( \Delta_g  - |\Lambda|, \left( \partial_n +  \frac{\alpha}{4}h \right)|_{\partial M} \right)$ is self-adjoint. Suppose that $u \in \ker P\subset C^{2,\beta}_{\delta} (M)$ with $\delta>1$, then
\begin{equation} \label{ker}
\begin{split}
0 = & \int_M \left(- u \Delta_g u + |\Lambda| u^2 \right)d V \\
& = \int_M \left(|\nabla u|^2_g + |\Lambda| u^2 \right)d V - \frac{\alpha}{4} \int_{\partial M } h u^2 d \sigma.
\end{split}
\end{equation}
We can also use the Sobolev trace inequality to find
\begin{equation}\label{jfjhwiw}
\int_M\left( |\nabla u|^2_g +  u^2 \right) d V \geq C \int_{\partial M} u^2 d \sigma,
\end{equation}
for some constant $C>0$ that will be examined below. By combining this with the previous equation we obtain
\begin{equation} \label{cconst}
0 \geq  \left( C\min\{1,|\Lambda|\}- \frac{1}{4} \max \displaylimits_{\partial M} h \right) ||u||^2_{H^1 \left(M \right)}.
\end{equation}
Therefore, the kernel is trivial if $\max \displaylimits_{\partial M} h < 4 C\min\{1,|\Lambda|\}$. In this case, the existence of a solution for $1 < \delta < 3$ follows from a similar argument as in \cite[Proposition 3.1]{maxwell}, adapted to the weighted H\"older spaces and Fredholm properties of elliptic operators in the asymptotically hyperbolic setting (\cite[Theorem C]{lee} with $R=2$). 

It remains to estimate the constant $C$ in \eqref{jfjhwiw}. Consider the distance function to the boundary $t (x)=\mathrm{dist}(x, \partial M)$, and let $T<T_0$ where $T_0$ is the normal injectivity radius of $\partial M$. Define a nonnegative cut-off function $\eta=\eta(t)$ such that $\eta(t)=1$ for $t\leq T/2$ and $\eta(t)=0$ for $t\geq T$, then
\begin{equation}
\begin{split}
& \int_{\partial M} u^2 d \sigma = \int_{\partial M} \eta^2 u^2 d \sigma = - \int_{M} \mathrm{div} \left(  \eta^2 u^2 \partial_t \right) d V \\
& = - \int_{M} \left(2 \eta u^2 \partial_t \eta + 2 u \eta^2 \partial_t u + u^2 \eta^2 h_t \right) d V \\
& \leq  \int_{M} \left( |\nabla u|_g^2 + \left(2 |\nabla \eta|_g + 1 + |h^-_t| \right) u^2 \right) d V \\
&  \leq  \left( \max \displaylimits_{M_T} \left(2 |\nabla \eta|_g + 1 + |h^-_t| \right) \right)  ||u||^2_{H^1 \left(M \right)}
\end{split}
\end{equation}
where $h_t$ is the mean curvature of the $t$-level sets and $h^-_t = \mathrm{min}\{h_t,0 \}$. Therefore we may take
\begin{equation}
C =  \left( \max \displaylimits_{M_T} \left(2 |\nabla \eta|_g + 1 + |h^-_t| \right) \right)^{-1},
\end{equation}
and the condition for a trivial kernel becomes
\begin{equation} \label{hineq}
\max \displaylimits_{\partial M} h < \frac{4 \min\{1,|\Lambda|\}}{ 8 T^{-1} + 1 + \max \displaylimits_{M_T} |h^-_t| },
\end{equation}
where the cut-off function has been chosen to ensure
\begin{equation}
\max \displaylimits_{M_T}  |\nabla \eta|_g \leq \frac{4}{T}.
\end{equation}
Moreover, using the lower bound on $|h^{-}_t|$ from Proposition \ref{meancurve} shows that (\ref{hineq}) holds if
\begin{equation} \label{hbound}
\max \displaylimits_{\partial M} h < \frac{4 T \min\{1,|\Lambda|\}}{ 8 + T +2  T^2  \tilde{\mathrm{K}}_1 \cosh^{2} ( \sqrt{|\mathrm{K}_0| } T )  },
\end{equation}
where
\begin{equation}
\tilde{\mathrm{K}}_1 := \max \displaylimits_{i=1,2} \max \displaylimits_{M_T} \mathrm{K}^+ \left(\partial_t, e_i  \right). 
\end{equation}
Since $T$ may be chosen arbitrarily within the injectivity radius, the desired result now follows.
\end{proof}

We are now in a position to use the sub/supersolution method of \cite[Proposition 2.1]{gicquaud} to solve the Lichnerowicz equation on an asymptotically hyperbolic manifold $(M,g)$ with constant negative scalar curvature $R_g = - 2 |\Lambda|$, and nonnegative mean curvature on the boundary. These conditions on the scalar and mean curvature may be assumed without loss of generality (if $0<\tau<3$) in light of \cite[Theorem 1.1]{gicquaud}, which shows that an arbitrary asymptotically hyperbolic manifold can be conformally transformed to achieve this outcome.

Our construction of the supersolution for the nonlinear boundary value problem (\ref{mthbvp1}) is motivated by a similar existence theorem given by Maxwell \cite{maxwell}. However, instead of assuming a global bound on the geometry in the form of a positive Yamabe invariant, we prove the existence of a supersolution under a condition bounding the mean curvature (\ref{existence_bound}). This approach also differs from the one presented by the Gicquaud \cite{gicquaud}, where the barriers are constructed in an explicit way.

In the next theorem we solve the Lichnerowicz equation with Robin boundary condition, which will be used to form initial data with a MOTS of positive mean curvature. This may be viewed as a complimentary result to the analogous theorem from \cite[Theorem 1.2]{gicquaud}, where the resulting MOTS has nonpositive mean curvature. The positive mean curvature property of the MOTS that we find here allows us to apply new initial data to study the asymptotically hyperbolic Penrose inequality in Section \ref{secpenr2}.

\begin{theorem} \label{mth}
Let $(M,g)$ be a $C^{2,\beta}_{\tau}$-asymptotically hyperbolic manifold with $0<\beta<1$, $1 < \tau < 3$, scalar curvature $R_g = - 2 |\Lambda|$, and boundary of nonnegative mean curvature $h$ that satisfies \eqref{existence_bound}. Assume that $k \in C^{1, \beta}_{\tau}(M)$ is a divergence and trace-free 2-tensor with $0 \leq k_{nn} \leq  h$ on $\partial M$. Then the boundary value problem
\begin{equation} \label{mthbvp}
\begin{split}
\Delta_g \phi &  = \frac{|\Lambda|}{4} \phi \left( \phi^4 -1 \right) - \frac{1}{8 } |k|^2_{g} \phi^{-7} \quad \mathrm{on} \quad M,
\\ \partial_n \phi &  = - \frac{1}{4}h \phi+ \frac{1}{4 }k_{nn}\phi^{-3}  \quad \mathrm{on} \quad \partial M,
\end{split}
\end{equation}
admits a positive solution with $\phi-1 \in C^{2,\beta}_{\tau}(M)$.
\end{theorem}

\begin{proof}
Let $\phi=1+v$ and set
\begin{equation}
\begin{split}
L \left( v \right) & := \Delta_g v - \frac{|\Lambda|}{4} \left(1+v \right) \left( \left(1+v  \right)^4 -1\right) + \frac{1}{8 \left(1+v \right)^7}|k|^2_{g}, \\
B \left( v \right) & :=\partial_n v + \frac{h}{4} \left(1+v \right) - \frac{k_{nn}}{4 \left(1+v \right)^3}.
\end{split}
\end{equation}
It can be checked that
\begin{equation}
L \left( 0 \right) = \frac{|k|^2_{g}}{8 }, \quad\quad\quad B \left( 0 \right) = \frac{h}{4} - \frac{k_{nn}}{4},
\end{equation}
so $v_- =0 $ is a subsolution due to the assumption that $k_{nn} \leq h$ on $\partial M$.

Following \cite{maxwell}, in order to find a supersolution consider the boundary value problem
\begin{equation} \label{supers}
\begin{split}
\Delta_g v_{\alpha} - |\Lambda|  v_{\alpha} &  = -  \frac{\alpha}{8}|k|_g^2 \quad \mathrm{on} \quad  M, \\
 \partial_n v_{\alpha} &  =  - \frac{\alpha}{4} h\left( 1+ v_{\alpha} \right) \quad \mathrm{on} \quad \partial M,
\end{split}
\end{equation}
for $\alpha \in [0,1]$. The existence of a solution $v_{\alpha} \in C^{2, \beta}_{\tau}(M)$ is guaranteed by Proposition \ref{existence}. Let $I= \{\alpha \in [0,1]\mid v_{\alpha} > -1 \}$. Notice that this set is nonempty since $v_0=0$, and from the estimates implicit in the proof of the previous proposition we find continuous dependence of the solutions on coefficients, showing that $I$ is open as well. Next consider $\alpha_0 \in \bar{I}$. Then $v_{\alpha_0} \geq -1$, as $\alpha_0$ is a limit point. In fact, the strong maximum principle implies that $v_{\alpha_0} > 0$, and therefore $I$ is closed. It follows that $I=[0,1]$, and we choose $v_+ = v_1$. Observe that
\begin{equation} \label{supsolc}
\begin{split}
L \left(v_+ \right) & = - \frac{|k|_g^2}{8} \left(1 - \frac{1}{\left(1+ v_+ \right)^7} \right) - \frac{|\Lambda|}{4} v_+^2 \left(10  + 10 v_+ + 5 v_+^2 + v_+^3  \right)\leq 0, \\
B \left(v_+ \right) & =  - \frac{k_{nn}}{4 \left(1+v_+ \right)^3}\leq 0,
\end{split}
\end{equation}
since $k_{nn} \geq 0$ on $\partial M$, and therefore $v_+$ is a supersolution. We can now use \cite[Proposition 2.1]{gicquaud} combined with the proof of \cite[Theorem 3.3]{gicquaud} to obtain a solution $\phi$ of (\ref{mthbvp}) with $\phi-1 \in C^{2,\beta}_{\tau}(M)$.
\end{proof}

\subsection{The momentum constraint}

The seed extrinsic curvature will be sought in the form $k=\mathcal{D}_g X$, where $X$ is a 1-form satisfying
\begin{equation} \label{mombvp}
\begin{split}
 \mathrm{div}_g \left( \mathcal{D}_g X \right) & = 0 \quad \mathrm{on} \quad M, \\
 \mathcal{D}_g X (n) & = \kappa \quad \mathrm{on} \quad \partial M.
\end{split}
\end{equation}
Note that the boundary data 1-form $\kappa$ has a normal component $\kappa(n) = k_{nn}$, which is required to satisfy an upper bound for applicability of Theorem \ref{mth}. In order to study the existence of solutions, consider $X \in C^{2,\beta}_{\delta} (M)$, $ \delta >1$ in the kernel of $( \mathrm{div}_g  \mathcal{D}_g, B|_{\partial M})$, where $B$ is the boundary operator from (\ref{mombvp}). Integrating by parts produces
\begin{equation}
0 = - \int_{M} \langle X, \mathrm{div}_g ( \mathcal{D}_g X) \rangle  d V = \frac{1}{2} \int_{M} |\mathcal{D}_g X|^2_g d V,
\end{equation}
so that $\mathcal{D}_g X =0$. This means that $X$ is a conformal Killing field vanishing at infinity, and it is a well-known fact that in the asymptotically hyperbolic setting there are no such nontrivial fields \cite{gicquaud, lee}. 
Moreover, the $L^2$-adjoint operator may be computed in a straightforward way. Let $ X, Y \in C^{2,\beta}_{\delta} (M)$ with $\delta >1$, then
\begin{align}
\begin{split}
\int_{M} \langle Y,\mathrm{div}_g (\mathcal{D}_g X)\rangle d V =& - \int_{\partial M}   \mathcal{D}_g X \left( Y, n \right) d \sigma - \int_{M}   \langle\mathcal{D}_g Y ,\nabla X \rangle d V \\
=& \int_{\partial M} \left( \mathcal{D}_g Y \left( X, n \right)- \mathcal{D}_g X \left( Y, n \right)\right) d \sigma + \int_{M} \langle\mathrm{div}_g( \mathcal{D}_g Y), X\rangle d V
\end{split}
\end{align}
showing that the operator is self-adjoint. Hence, the boundary value problem (\ref{mombvp}) has trivial cokernel, and may be solved for arbitrary $\kappa$ as long as the operator is Fredholm.  A complete proof, with general fall-off, is presented in \cite[Theorem 6.9]{gicquaud}.

\begin{theorem}\label{akfhoiuqoih}
Let $(M,g)$ be a $C^{2,\beta}_{\tau}$-asymptotically hyperbolic manifold  with $0<\beta<1$, $\tau >0$. Let $\kappa$ be a smooth section of $T^* M$ restricted to $\partial M$, then the boundary value problem
\begin{equation}
\begin{split}
 \mathrm{div}_g \left( \mathcal{D}_g X \right) & = 0 \quad \mathrm{on} \quad M, \\
 \mathcal{D}_g X (n) & = \kappa \quad \mathrm{on} \quad \partial M,
\end{split}
\end{equation}
admits a unique solution $X \in C^{2, \beta}_{\delta}(M)$ for any $1<\delta<3$.
\end{theorem}

\section{Initial Data With a MOTS Conformal to Schwarzschild-AdS} \label{secpenr}

The conformal method of the Einstein constraint equations can be used to verify the Penrose inequality for a class of perturbations of known black hole solutions. In \cite{kt1, kt2}, this has been carried out for the Penrose inequality with angular momentum assuming  axisymmetric perturbations of Schwarzschild initial data. As an application of the existence results derived above for the constraints with a MOTS boundary in the asymptotically hyperbolic setting, we will study the hyperbolic Penrose inequality for perturbations of Schwarzschild-AdS initial data. In this section we apply Theorem \ref{mth} to construct initial data conformal to Schwarzschild-AdS as well as calculate mass and linear momentum for such data. In the following section we will treat the Penrose inequality for perturbations of Schwarzschild-AdS initial data.

\subsection{Schwarzschild-AdS initial data}
The Schwarzschild-AdS metric in Schwarzschild coordinates is given by
\begin{equation} \label{SchwAdS}
\mathfrak{g}= - \left( 1-\frac{2m}{\rho} + \frac{|\Lambda|}{3} \rho^2 \right) dt^2 + \frac{d\rho^2}{1-\frac{2m}{\rho} + \frac{|\Lambda|}{3} \rho^2} + \rho^2 g_{S^2},
\end{equation}
on the domain of outer communication $\mathbb{R}\times [\rho_h ,\infty)\times S^2$.
The first fundamental form of a $t=$const hypersurface is then
\begin{equation} \label{SchwAdSinihyp}
g_{S} =  \frac{d\rho^2}{1-\frac{2m}{\rho} + \frac{|\Lambda|}{3} \rho^2} + \rho^2 g_{S^2},
\end{equation}
and the radius $\rho_h$ at which a MOTS (minimal surface) occurs satisfies the relation
\begin{equation} \label{mrhoh}
\theta_+ =h= 0 \iff m = \frac{\rho_h}{2} \left(1+ \frac{|\Lambda|}{3} \rho_h^2 \right).
\end{equation}
Note that that there is a single positive real root (MOTS radius) regardless of the values of $m> 0$ and $|\Lambda|$.

\subsection{The total energy-momentum vector of asymptotically hyperbolic initial data}
Let $(M, g, k)$ be an asymptotically hyperbolic vacuum initial data set of class $(\beta, \tau)$ with $\tau > \frac{3}{2}$. Such initial data have a well-defined total energy and linear momentum (see Michel \cite{michel}) defined as follows. Let
\begin{equation}
b = \frac{d \rho^2}{1+ \frac{| \Lambda|}{3} \rho^2} + \rho^2 g_{S^2}
\end{equation}
be the hyperbolic reference metric in hyperboloidal coordinates, which may be obtained from \eqref{refmet} by setting
$\rho = \sqrt{\frac{3}{|\Lambda|}} \sinh r$. Define $\bar{g} = \Psi_* g - b$
where $\Psi$ is the diffeomorphism defining the asymptotic coordinate system,
and consider lapse functions for Killing fields on AdS
\begin{equation}
V_0 = \sqrt{1+ \frac{| \Lambda|}{3} \rho^2}, \quad\quad\quad V_i =  \sqrt{\frac{|\Lambda|}{3}} x^i \rho,
\end{equation}
where $x^i$ are the Cartesian coordinates restricted to the unit sphere $S^2$, that is
\begin{equation}
x^1 = \sin \theta \cos \varphi, \quad x^2 = \sin \theta \sin \varphi, \quad x^3 = \cos \theta.
\end{equation}
We have
\begin{equation}
\Hess_b V_{\alpha} = \frac{|\Lambda|}{3} b V_{\alpha}.
\end{equation}
Following \cite{Cederbaum, Chrusciel_Maerten_Tod}, the total energy-momentum vector of $(M, g, k)$ can be defined as
\begin{equation} \label{enemomvec}
p_{\alpha}  := \frac{1}{16 \pi }  \lim_{\rho \to \infty}\int \displaylimits_{S_{\rho}}\left[   V_{\alpha} \left( \mathrm{div}_b \bar{g} \right)
 - V_{\alpha} \left( d \tr_b \bar{g} \right)   +  \tr_b \bar{g} \left( d V_{\alpha} \right)  -  \bar{g}  \left(\nabla_b V_{\alpha} \right) \right] \left(\nu_b \right) d \sigma_b,
\end{equation}
where $\alpha \in \{0,1,2,3\}$ and $\nu_b = \sqrt{1+\frac{|\Lambda|}{3} \rho^2 } \partial_{\rho}$. The mass $\mathfrak{m}$ is then defined with a Lorentzian norm of the energy-momentum 4-vector
\begin{equation}
\mathfrak{m}^2 := p_0^2 - \sum_{i=1}^3 p_i^2.
\end{equation}
In particular, if the metric is conformal to Schwarzschild-AdS, that is $g=\phi^4 g_{S}$, then
\begin{align}
\begin{split}
p_0 &= m - \frac{1}{8 \pi} \lim_{\rho \to \infty}\int \displaylimits_{S_{\rho}} \left[ 4 \phi^3 \left(1+\frac{|\Lambda|}{3} \rho^2\right) \partial_{\rho} \phi - \frac{|\Lambda|}{3} \rho \left(  \phi^4- 1 \right) \right] d \sigma, \\
p_i & =  - \frac{1}{8 \pi}   \lim_{\rho \to \infty}\int \displaylimits_{S_{\rho}} \left(\frac{1}{2 \rho }   + \frac{|\Lambda|}{3} \rho \right) x^i \left[ 4  \rho \phi^3  \partial_{\rho} \phi   -    \left( \phi^4 - 1 \right)  \right] d \sigma,
\end{split}
\end{align}
where the energy-momentum vector of Schwarzschild-AdS initial data is
\begin{equation}
p_{0}^{S} = m,\quad \quad p_i^{S} = 0.
\end{equation}
Moreover, we have
\begin{equation}
m  = \left( \frac{A_h^{S}}{16 \pi} \right)^{1/2} + \frac{4}{3}  |\Lambda| \left( \frac{A_h^{S}}{16 \pi} \right)^{3/2},
\end{equation}
where $A_h^{S} = 4 \pi \rho_h^2$ is the area of the minimal surface within the constant time slice of the Schwarzschild-AdS spacetime.

\subsection{Perturbations of Schwarzschild-AdS initial data}\label{6.3}

We shall apply Theorem \ref{mth} to construct a large class of perturbed Schwarzschild-AdS data having a MOTS boundary with positive mean curvature, and will then show that the Penrose inequality holds for sufficiently small perturbation parameter. It should be noted, however, that Theorem \ref{mth} cannot be directly applied to the time slice of Schwarzschild-AdS $(M_S, g_S)$ as this data possesses a minimal boundary. More precisely, we require the condition that $0 \leq k_{nn} \leq  h$ where $h$ is the boundary mean curvature of the seed data. Therefore, we will use the hyperbolic metric $b$ as a seed, together with the fact that the time slices of the Schwarzschild-AdS spacetime are conformal to hyperbolic space, that is $g_S=\phi^4_S b$. 

Suppose that the asymptotically hyperbolic vacuum initial data with a MOTS $(M, \hat{g} , \hat{k})$ is constructed via the conformal method of Theorem \ref{mth} with seed metric $b$ and seed extrinsic curvature $k^b$, that is
\begin{equation}
M=[r_h,\infty)\times S^2,\quad\quad \hat{g} = \bar{\phi}^4 b, \quad\quad \hat{k} = \bar{\phi}^{-2} k^b,
\end{equation}
where $\bar{\phi}$ is a solution of the Lichnerowicz equation with the MOTS boundary condition (\ref{mthbvp}) and $k^b$ is a divergence-free and trace-free tensor with respect to $b$. The hypotheses of Theorem \ref{mth} require that seed data $(b, k^b)$ satisfy
\begin{equation} \label{hyperconf}
0 \leq k_{nn}^b \leq \sqrt{ \frac{4 |\Lambda|}{3}} \coth r < 4\min\{1, |\Lambda|\},
\end{equation}
where the quantity after the second inequality is the mean curvature of the radius $r$-sphere in hyperbolic space with respect to geodesic polar coordinates (see (\ref{refmet})). Moreover, we used the fact that the sectional curvature of hyperbolic space is negative in the last inequality which arises from (\ref{existence_bound}). We would like the data $(M, \hat{g}, \hat{k})$ to be expressed as conformal to Schwarzschild-AdS initial data with minimal surface inner boundary $\partial M$, that is
\begin{equation} \label{schwadsconf}
\hat{g}  = \phi^4 g_{S}, \quad\quad \hat{k} = \phi^{-2} k , 
\end{equation}
where $\phi = \bar{\phi} \phi_S^{-1}$ and $k= \phi^{-2}_S k^b$. Notice that $k$ is divergence-free and trace-free with respect to $g_S$, and $\phi$ satisfies the following boundary value problem
\begin{align} \label{lichschwads}
\begin{split}
\Delta_{g_S} \phi &  = \frac{|\Lambda|}{4} \phi \left( \phi^4 -1 \right) - \frac{1}{8} |k|^2_{g_S}  \phi^{-7} \quad \mathrm{on} \quad M,
\\ \partial_n \phi &  =  \frac{1}{4 } k_{nn} \phi^{-3}  \quad \mathrm{on} \quad \partial M, \\
\end{split}
\end{align}
if $\partial M$ corresponds to the minimal surface of the Schwarzschild-AdS constant time slice. Let $m$ and $\Lambda$ be the mass and cosmological constant of the Schwarzschild-AdS data, and denote by $r_h = r_h (m, \Lambda)$ the (geodesic polar) radial coordinate of the conformal hyperbolic space which corresponds to the horizon. 
Then, in order to solve \eqref{lichschwads} with Theorem \ref{mth}, we require $r_h$ to satisfy (\ref{hyperconf}). We now show that there is an open set within the range of parameters $(m,\Lambda)$ for which (\ref{hyperconf}) is satisfied.

The Schwarzschild-AdS canonical slice metric may be expressed in Schwarzschild coordinates, and conformal to hyperbolic space with geodesic polar coordinates, to find
\begin{equation} 
\frac{d\rho^2}{1-\frac{2m}{\rho} + \frac{|\Lambda|}{3} \rho^2} + \rho^2 g_{S^2} =  \frac{3}{|\Lambda|} \phi_S^4 \left( d r^2 + \sinh^2 r g_{S^2} \right).
\end{equation}
Therefore 
\begin{equation} \label{drhodr}
\phi_S^2 = \sqrt{\frac{|\Lambda|}{3}} \frac{\rho}{\sinh r},\quad\quad\quad\frac{d \rho}{d r} = \frac{\sqrt{\rho^2 - 2 m \rho + \frac{|\Lambda|}{3} \rho^4}}{\sinh r}.
\end{equation}
Below, in Lemma \ref{lemmarh}, we use these relations to show that $r_h (m, \Lambda) = \infty$ as $m\rightarrow\infty$. In this limit, the mean curvatures of the coordinate spheres in hyperbolic space approach a constant value
\begin{equation}
\lim \displaylimits_{m \to \infty} \sqrt{ \frac{4|\Lambda|}{3}} \coth r_h = \sqrt{ \frac{4|\Lambda|}{3}}.
\end{equation}
Thus, the last inequality of (\ref{hyperconf}) is valid for the horizon if $\frac{1}{12} < |\Lambda| < 12$ for sufficiently large $m$. We conclude that Theorem \ref{mth} may be used to construct initial data conformal to the full exterior region of Schwarzschild-AdS, with parameters $m$ and $\Lambda$ in this range, such that the boundary of the new data is a MOTS with positive mean curvature.

\begin{lemma}\label{lemmarh}
Let $r_h(m,\Lambda)$ denote the radial geodesic polar coordinate of hyperbolic space, which corresponds to the horizon in the conformal Schwarzschild-AdS time slice with mass $m$ and cosmological constant $\Lambda$. 
Then 
\begin{equation}
\lim \displaylimits_{m \to \infty} r_h (m, \Lambda) = \infty.
\end{equation}
\end{lemma}

\begin{proof}
Let $\rho_h(m,\Lambda)$ be the radius of the horizon in Schwarzschild coordinates, as in (\ref{mrhoh}). Then equation (\ref{drhodr}) implies
\begin{equation}
\int_{r_h}^{\infty} \frac{d \tilde{r}}{\sinh \tilde{r}} = \int_{\rho_h}^{\infty} \frac{d \tilde{\rho} }{ \sqrt{\tilde{\rho}^2 -2 m \tilde{\rho } + \frac{|\Lambda|}{3} \tilde{\rho}^4} }.
\end{equation}
Using the change of variables $\bar{\rho} = \tilde{\rho} m^{-1/3}$, we find that
\begin{equation}
- \log \left( \tanh \frac{r_h}{2} \right)= m^{-1/3} \int_{\rho_h m^{-1/3}}^{\infty} \frac{d \bar{\rho} }{ \sqrt{  m^{-2/3} \bar{\rho}^2 -2 \bar{\rho } + \frac{|\Lambda|}{3} \bar{\rho}^4} }.
\end{equation}
Furthermore, from (\ref{mrhoh}) it can be shown that the asymptotic expansion of $\rho_h$ for large $m$ is given by
\begin{equation}
\rho_h = \left( \frac{6 m}{|\Lambda|} \right)^{1/3} + O(m^{-1/3}).
\end{equation}
Therefore
\begin{equation}
  \lim_{m \to \infty}   \bigg| \log \left(  \tanh \frac{r_h}{2} \right)  \bigg| \leq \lim_{m \to \infty} \left( m^{-1/3} \int_{\left(\frac{6}{|\Lambda|} \right)^{1/3}}^{\infty} \frac{d \bar{\rho} }{ \sqrt{ \frac{|\Lambda|}{3} \bar{\rho}^4  -2 \bar{\rho } } } \right)=0.
\end{equation}
It follows that $r_h \rightarrow\infty$ as $m\rightarrow\infty$.
\end{proof}

\section{The Penrose inequality for perturbations of Schwarzschild-AdS initial data} \label{secpenr2}


\subsection{Setup}
In the previous section we showed how to apply the conformal method to construct asymptotically hyperbolic initial data $(M,\hat{g}=\phi^4 g_S,\hat{k}=\phi^{-2} k)$ with MOTS boundary from the seed $(g_S,k)$, where $g_S$ arises from a constant time slice of the Schwarzschild-AdS spacetime of mass $m$ and cosmological constant $\Lambda$, and $k$ is a transverse-traceless tensor with respect to $g_S$. This data is maximal and satisfies the vacuum constraints, in particular $\phi$ solves \eqref{lichschwads}. We will now consider data of this form which are perturbations of Schwarzschild-AdS. More precisely, it will be assumed that
\begin{equation} \label{fact_exp}
\phi = 1 + \sum \displaylimits_{j = 1}^{\infty}  \phi_j \epsilon^j,
\quad\quad k = \epsilon k',
\end{equation}
for some perturbation parameter $\epsilon>0$, where $\phi_{j} \in C^{2,\beta}_{\tau}(M)$, and $k'\in C^{1,\beta}_{\tau}(M)$ is again transverse-traceless with respect to $g_S$, such that $\tau>3/2$.  The hyperbolic Penrose inequality takes the form
\begin{equation} \label{pineq1}
\mathfrak{m}=\sqrt{p_0^2 -\sum \displaylimits_{i=1}^3 p_i^2 } \geq \mathcal{A},
\end{equation}
where
\begin{equation}
\mathcal{A}  = \sqrt{\frac{A_h}{16 \pi}} \left(1+ \frac{ |\Lambda| }{12 \pi} A_h  \right)
\end{equation}
and $A_h$ is the minimal area required to enclose the MOTS boundary. In the sequel we will verify the Penrose inequality in which $A_h$ is instead taken to be the area of the inner boundary MOTS, and will later show that in the current perturbed context these two areas agree if the boundary mean curvature is positive.

Each quantity appearing in the Penrose inequality may be expanded in the perturbation parameter $\epsilon$ as follows
\begin{equation}
p_0 =  m +  \sum \displaylimits_{j=1}^{\infty}  p_0^{(j)} \epsilon^j , \quad\quad\quad p_i =  \sum \displaylimits_{j=1}^{\infty}  p_i^{(j)} \epsilon^j, \quad\quad\quad \mathcal{A}  = m +  \sum \displaylimits_{j=1}^{\infty} \mathcal{A}^{(j)} \epsilon^j.
\end{equation}
We will study the Penrose inequality up to second order of the expansion. In order to compute the corrections to the energy $p_0$ we will make use of the following identity
\begin{align} \label{enecomp}
\begin{split}
 & \lim \displaylimits_{\rho \to \infty} \rho^2 \left[ \left(1- \frac{2m}{\rho} + \frac{|\Lambda|}{3} \rho^2  \right) \partial_{\rho} \left< \phi \right> - \frac{|\Lambda|}{3} \rho \left< \phi -1\right> \right]\\
= & -\frac{\rho_h}{2} \left< \phi -1\right>_h \left( 1 +|\Lambda| \rho_h^2 \right) -\frac{1}{8} \int \displaylimits_{\rho_h}^{\infty} \left< \frac{|k|^2_{g_S}}{\phi^7}  \right> \rho^2 d \rho + \int \displaylimits_{\rho_h}^{\infty} \frac{|\Lambda|}{4} \rho^2 \left(\left< \phi^5 \right> - 5 \left< \phi \right> +4 \left< 1\right> \right) d \rho,
\end{split}
\end{align}
where $\langle \cdot \rangle$ denotes integration over the unit sphere 
and $\langle \cdot \rangle_h$ is the same integral with the integrand evaluated at the inner boundary ($\rho=\rho_h$). This identity can be derived from
the Laplacian expressed in Schwarzschild coordinates
\begin{align}
\begin{split}
\Delta_{g_S} \phi & = \frac{  \sqrt{1-\frac{2m}{\rho} + \frac{|\Lambda|}{3} \rho^2}}{ \rho^2} \partial_{\rho} \left(  \rho^2 \sqrt{1-\frac{2m}{\rho} + \frac{|\Lambda|}{3} \rho^2} \partial_{\rho} \phi \right) + \frac{1}{\rho^2} \Delta_{S^2} \phi \\
     & = \frac{1}{\rho^2} \partial_{\rho} \left[  \rho^2 \left( 1-\frac{2m}{\rho} + \frac{|\Lambda|}{3} \rho^2 \right) \partial_{\rho} \phi  \right] - \frac{1}{\rho^2} \left( m + \frac{|\Lambda|}{3} \rho^3 \right) \partial_{\rho} \phi  + \frac{1}{\rho^2} \Delta_{S^2} \phi \\
     & = \frac{1}{\rho^2} \partial_{\rho} \left[  \rho^2 \left( 1-\frac{2m}{\rho} + \frac{|\Lambda|}{3} \rho^2 \right) \partial_{\rho} \phi - \left( m + \frac{|\Lambda|}{3} \rho^3 \right) \phi  \right] + |\Lambda| \phi + \frac{1}{\rho^2} \Delta_{S^2} \phi,
\end{split}
\end{align}
by integrating with respect to the `flat volume form' $\rho^2 \sin\theta d\rho d\theta d\varphi$, and employing the Lichenrowicz equation (\ref{lichschwads}) as well as the relation (\ref{mrhoh}) between the mass $m$ and the radius of inner boundary $\rho_h$.  Up to the second order of expansion, the energy $p_0$ reads
\begin{equation} 
\begin{split}
p_0 = & m  - \frac{1}{2 \pi} \lim \displaylimits_{\rho \to \infty} \rho^2 \left[\left( 1 + \frac{|\Lambda|}{3} \rho^2 \right) \left(  \partial_{\rho} \left< \phi_1 \right>  \epsilon + \partial_{\rho} \left< \phi_2 \right> \epsilon^2  \right) - \frac{|\Lambda|}{3} \rho \left(   \left<\phi_1 \right> \epsilon + \left< \phi_2 \right> \epsilon^2 \right) \right] \\
 & -\frac{3}{2 \pi }   \lim \displaylimits_{\rho \to \infty} \rho^2   \left(1 + \frac{|\Lambda|}{3} \rho^2  \right) \left< \phi_1 \partial_{\rho}\phi_1 \right> \epsilon^2  + \frac{|\Lambda|}{4 \pi }  \lim \displaylimits_{\rho \to \infty} \rho^3  \left< \phi_1^2\right> \epsilon^2   + O \left( \epsilon^3\right).
\end{split}
\end{equation}
The first line of the expression above can be simplified with the use of (\ref{enecomp}), namely
\begin{equation} \label{p0e2}
\begin{split}
 p_0  = & m + \frac{\rho_h}{4 \pi} \left( \left<  \phi_1 \right>_h \epsilon +\left< \phi_2 \right>_h \epsilon^2  \right) \left( 1 +|\Lambda| \rho_h^2 \right) + \left( \frac{1}{16 \pi} \int \displaylimits_{\rho_h}^{\infty} \left< |k'|^2_{g_S} \right> \rho^2 d \rho -  \frac{ 5 |\Lambda|}{4 \pi} \int \displaylimits_{\rho_h}^{\infty} \left<\phi_1^2 \right> \rho^2 d \rho \right) \epsilon^2 \\
  & -\frac{3}{2 \pi }  \lim \displaylimits_{\rho \to \infty} \rho^2   \left(1 + \frac{|\Lambda|}{3} \rho^2  \right) \left< \phi_1 \partial_{\rho}\phi_1 \right> \epsilon^2   + \frac{|\Lambda|}{4 \pi }  \lim \displaylimits_{\rho \to \infty} \rho^3  \left< \phi_1^2\right> \epsilon^2   + O \left( \epsilon^3\right).
\end{split}
\end{equation}
Moreover, because the first correction to the conformal factor $\phi_1$ satisfies
\begin{equation} \label{phi1eq}
 \begin{split}
 \Delta_{g_S} \phi_1 - |\Lambda| \phi_1  =0, \quad \mathrm{on} \quad M, \\
 \partial_n \phi_1 - \frac{k'_{nn}}{4} = 0, \quad \mathrm{on} \quad \partial M,
 \end{split}
\end{equation}
we have $\phi_1 = O_1\left(\rho^{-3} \right)$ in the asymptotic region. Hence, the last two terms in (\ref{p0e2}) vanish so that
\begin{equation} \label{p0e21}
\begin{split}
 p_0  = & m + \frac{\rho_h}{4 \pi} \left( \left<  \phi_1 \right>_h \epsilon +\left< \phi_2 \right>_h \epsilon^2  \right) \left( 1 +|\Lambda| \rho_h^2 \right) \\
  & + \left( \frac{1}{16 \pi} \int \displaylimits_{\rho_h}^{\infty} \left< |k'|^2_{g_S} \right> \rho^2 d \rho -  \frac{ 5 |\Lambda|}{4 \pi} \int \displaylimits_{\rho_h}^{\infty} \left<\phi_1^2 \right> \rho^2 d \rho \right) \epsilon^2 + O \left( \epsilon^3\right).
\end{split}
\end{equation}
The surface area term $\mathcal{A}$ up to the second order of expansion reads
\begin{equation} \label{s1}
\mathcal{A}^{(1)} =  \frac{\rho_h}{4 \pi} \left< \phi_1 \right>_h  \left( 1+ |\Lambda|  \rho_h^2 \right),
\end{equation}
and
\begin{equation} \label{s2}
 \mathcal{A}^{(2)}   =  \rho_h  \left( 1+ |\Lambda| \rho_h^2  \right) \left(  \frac{1}{4 \pi }\left< \phi_2 \right>_h   + \frac{3 }{8 \pi} \left< \phi_1^2 \right>_h \right)  -  \frac{\rho_h}{16 \pi^2} \left< \phi_1 \right>_h^2 \left( 1 - |\Lambda| \rho_h^2 \right),
\end{equation}
whereas the first correction to the linear momentum is
\begin{equation} \label{pi1}
p_i^{(1)}    = -  \frac{|\Lambda|}{6 \pi}  \lim_{\rho \to \infty} \rho^3 \left< x^i \rho \partial_{\rho} \phi_1 - x^i \phi_1  \right>.
\end{equation}
It is worth noticing that in contrast to the asymptotically flat setting considered in \cite{kt1, kt2}, the first contribution to the linear momentum does not depend on the extrinsic curvature. This is in accordance with the definition (\ref{enemomvec}).

In the first order of expansion
\begin{equation} \label{firstordm}
    \mathfrak{m}^{(1)} = \left( \sqrt{p_0^2 - \sum \displaylimits_{i=1}^3 p_i ^2} \right) ^{(1)} = p_0^{(1)} = \frac{\rho_h}{4 \pi}  \left< \phi_1 \right>_h \left( 1 +|\Lambda| \rho_h^2 \right).
\end{equation}
Thus, after comparing (\ref{s1}) and (\ref{firstordm}) we find that the Penrose inequality is saturated in the first order of expansion. Next observe that (\ref{p0e21}) implies
\begin{equation}
    p_0^2 = m^2 + 2m  p_0^{(1)} \epsilon +  \left( p_0^{(1)} \right)^2 \epsilon^2 + 2m  p_0^{(2)} \epsilon^2 + O(\epsilon^3), \quad\quad\quad p_i^2 =   \left( p_i^{(1)} \right)^2 \epsilon^2 + O(\epsilon^3),
\end{equation}
and therefore in the second order
\begin{equation}
\mathfrak{m}^{(2)} =  \left( \sqrt{p_0^2 - \sum \displaylimits_{i=1}^3  p_i^2} \right)^{(2)}   = p_0^{(2)} - \frac{1}{2m} \sum_{i=1}^3 \left( p_i^{(1)} \right)^2.
\end{equation}
It follows that the Penrose inequality is valid up to second order if
\begin{equation} \label{PI2}
p_0^{(2)} - \frac{1}{2m} \sum_{i=1}^3 \left( p_i^{(1)} \right)^2 \geq \mathcal{A}^{(2)}.
\end{equation}
In order to simplify this expression we will make use of the following estimate
\begin{equation} \label{phi1volest}
 | \Lambda|  \int \displaylimits_M   \phi_1^2  \sqrt{1-\frac{2m}{\rho} + \frac{|\Lambda|}{3} \rho^2 } d V_S \leq  \frac{\rho_h}{2} \left< \phi_1^2 \right>_h \left( 1 + |\Lambda| \rho_h^2 \right),
\end{equation}
where $dV_S$ is the volume form of the Schwarzschild-AdS metric (\ref{SchwAdSinihyp}).
The inequality \eqref{phi1volest} may be obtained by first multiplying equation (\ref{phi1eq}) with $\phi_1  \sqrt{1-\frac{2m}{\rho} + \frac{|\Lambda|}{3} \rho^2 }$ and integrating by parts
\begin{equation} \label{phi1volest2}
\frac{1}{2} \int \displaylimits_M \left< \nabla \sqrt{1-\frac{2m}{\rho} + \frac{|\Lambda|}{3} \rho^2 } , \nabla \phi_1^2 \right> dV_S + |\Lambda| \int \displaylimits_M  \phi_1^2 \sqrt{1-\frac{2m}{\rho} + \frac{|\Lambda|}{3} \rho^2 } dV_S
\leq 0,
\end{equation}
together with the calculation
\begin{align} \label{phi1volest3}
\begin{split}
& \frac{1}{2} \int \displaylimits_M \left< \nabla \sqrt{1-\frac{2m}{\rho} + \frac{|\Lambda|}{3} \rho^2 } , \nabla \phi_1^2 \right> dV_S\\
=& \frac12 \int \displaylimits_{\rho_h}^{\infty} \left(m + \frac{|\Lambda|}{3} \rho^3 \right) \partial_{\rho } \left< \phi_1^2 \right> d \rho\\
= & - \frac{\rho_h}{4} \left< \phi_1^2 \right>_h \left( 1 + |\Lambda| \rho_h^2 \right) - \frac{1}{2} |\Lambda| \int \displaylimits_M \phi_1^2 \sqrt{1-\frac{2m}{\rho} + \frac{|\Lambda|}{3} \rho^2 }  dV_S,
\end{split}
\end{align}
where $\phi_1 = O\left(\rho^{-3} \right)$ and (\ref{mrhoh}) have been used. 

Ultimately, with the use of (\ref{p0e2}), (\ref{s2}), (\ref{pi1}) and the estimate (\ref{phi1volest}), the Penrose inequality in the second order of expansion holds if the following inequality is satisfied
\begin{align} \label{PIe2}
\begin{split}
&\frac{1}{4 } \int \displaylimits_M  |k'|^2_{g_S} \sqrt{1-\frac{2m}{\rho} + \frac{|\Lambda|}{3} \rho^2 }  d V_S + \frac{\rho_h}{4 \pi} \left< \phi_1 \right>_h^2 \\
\geq & \frac{|\Lambda|^2}{18 \pi m } \sum_{i=1}^3 \left( \lim_{\rho \to \infty} \rho^3 \left< x^i \rho \partial_{\rho} \phi_1 - x^i \phi_1  \right> \right)^2  
 + \frac{\rho_h^3}{4 \pi} |\Lambda| \left< \phi_1 \right>_h^2 + 4 \rho_h   \left< \phi_1^2 \right>_h  \left( 1 + |\Lambda| \rho_h^2  \right).
\end{split}
\end{align}
Verification of a similar inequality in the asymptotically flat case 
was accomplished by finding an explicit expression for $\phi_1$ in terms of its Neumann boundary data, see \cite{kt1, kt2}. However, in the current asymptotically hyperbolic regime such an explicit expression is no longer feasible. Thus, in the next subsection we will proceed by making appropriate estimates instead of explicit computations.

\subsection{Verifying the Penrose inequality up to second order of expansion}

In this subsection we will show that if the normal component $k_{nn}= \epsilon k'_{nn}$ of the seed extrinsic curvature has small $H^1$-norm in a neighborhood of the boundary $\partial M$ compared to a global weighted $L^2$-norm of $|k|_{g_S}$, then inequality (\ref{PIe2}) will be satisfied. The first step to achieve this goal is to show that $\phi_1$ is globally controlled by $k'_{nn}$ near the boundary.

\begin{lemma} \label{lemmap1knn}
Let $\phi_1 \in C^{2, \beta}_{3}(M)$ be the solution of \eqref{phi1eq}.
Then
\begin{equation} \label{phi1knn}
||\phi_1||_{H^1 \left( M \right)} \leq C ||k'_{nn}||_{H^1 \left( M' \right)},
\end{equation}
where $M'$ is any precompact domain within $M$ containing the boundary $\partial M$ and $C = C(M')$.
\end{lemma}

\begin{proof}
The boundary value problem (\ref{phi1eq}) can be used to derive the following equality,
\begin{equation} \label{phi1less0}
\int \displaylimits_{M} \left( |\nabla \phi_1|^2_{g_S} + |\Lambda| \phi_1^2  \right) dV_S = - \frac{1}{4} \int \displaylimits_{\partial M}   \phi_1 k'_{nn} d \sigma_S.
\end{equation}
Let $\eta \in C^{\infty}_c (M)$ be a nonnegative cut-off function such that $\eta \equiv 1 $ on some precompact domain $M' \subset M$ containing the boundary $\partial M$. Let $t:M\rightarrow\mathbb{R}_{+}$ denote the distance function to $\partial M$, and note that since the manifold is a Schwarzschild-AdS time slice the injectivity radius of the boundary is infinite; hence the function $t$ is globally smooth. 
We have
\begin{equation}\label{knnineq}
 -  \int \displaylimits_{\partial M}   \phi_1 k'_{nn} d \sigma_S  = \int \displaylimits_{M'} \partial_t \left( \eta \phi_1 k'_{nn}\right) d V_S 
\leq  \gamma||\phi_1||^2_{H^1 \left( M' \right)} + C_1 \gamma^{-1}  ||k'_{nn}||^2_{H^1 \left( M' \right)},
\end{equation}
where Young's inequality has been used and $\gamma>0$ is a parameter. By choosing $\gamma$ sufficiently small, the combination of (\ref{phi1less0}) and (\ref{knnineq}) yields
\begin{equation} \label{phi1knnnorm}
||\phi_1||_{H^1 \left( M \right)} \leq C_2 ||k'_{nn}||_{H^1 \left( M' \right)},
\end{equation}
where $C_2$ depends on the choice of $M'$.
\end{proof}

This lemma shows that the last two terms on the right-hand side of (\ref{PIe2}) can be estimated in terms of the squared $H^1$-norm of $k_{nn}$ in the neighborhood $M'$ of the boundary $\partial M$. We will now estimate the first term on the right-hand side of (\ref{PIe2}) in terms of the same quantity. 

\begin{lemma} \label{lemap1approx}
Let $\phi_1 \in C^{2, \beta}_{3}(M)$ be the solution of \eqref{phi1eq}. Then
\begin{equation} \label{phi1linmom}
\lim_{\rho \to \infty} \rho^3 \big| \left< x^i \rho \partial_{\rho} \phi_1 - x^i \phi_1  \right> \big| \leq C_3 ||k'_{nn}||_{H^1 \left( M' \right)},
\end{equation}
where $C_3$ depends on the choice of $M'$.
\end{lemma}

\begin{proof} 
In order to estimate the limit in (\ref{phi1linmom}) we will make use of the asymptotic expansion of the equation from (\ref{phi1eq}), that is
\begin{equation} \label{asymptphi1}
\begin{split}
&  \frac{|\Lambda|}{3 \rho} \partial_{\rho} \left( \rho^3 \partial_{\rho} \phi_1 \right) + \frac{1}{\rho^2} \Delta_{S^2} \phi_1 - |\Lambda| \phi_1 = O(\rho^{-3}) \quad \mathrm{on} \quad M \setminus B_{\overline{\rho}}(0) ,
\end{split}
\end{equation}
where the radius $\overline{\rho}$ is chosen sufficiently large. In this region the solution may be represented using (real form) spherical harmonics
\begin{equation}
\phi_1 (\rho,\theta,\varphi)= \sum_{\ell=0}^{\infty} \sum_{j=- \ell}^{\ell} c_{\ell j} \mathcal{R}_{\ell} (\rho) Y_{\ell j} (\theta, \varphi),
\end{equation}
where the leading term in each $\mathcal{R}_{\ell}$ is given in terms of a modified Bessel function of the first kind
\begin{equation}
\mathcal{R}_{\ell}(\rho) \sim \frac{a_{\ell}}{\rho} I_2 \left( \frac{b_{\ell}}{\rho} \right) = O (\rho^{-3}),
\end{equation}
for some constants $a_\ell$, $b_\ell$, and $c_{\ell j}$. If we assume that $\mathcal{R}_{\ell}(\bar{\rho})=1$, then expanding $\phi_1$ on the inner boundary yields
\begin{equation}
    \phi_1 \big|_{\partial B_{\bar{\rho}}(0)} =  \sum_{\ell=0}^{\infty} \sum_{j=-\ell}^{\ell} c_{\ell j} Y_{\ell j}.
\end{equation}
It follows that
\begin{equation} \label{phi1limb}
\lim_{\rho \to \infty} \rho^3 \big|\left< x^i \rho \partial_{\rho} \phi_1 - x^i \phi_1  \right> \big| \leq C \sum_{j=-1}^1 |c_{1j}|,
\end{equation}
since the $\ell=1$ spherical harmonics are normalized restrictions of the Cartesian coordinates to the unit sphere.
The desired estimate \eqref{phi1linmom} may now be obtained by observing that the right-hand side of (\ref{phi1limb}) can be estimated in terms of the $H^1(M)$ norm of $\phi_1$ by the Sobolev trace theorem, and this norm can in turn be estimated by $||k'_{nn}||_{H^1 \left( M' \right)}$ according to 
Lemma \ref{lemmap1knn}.
\end{proof}


By Lemmas \ref{lemmap1knn} and \ref{lemap1approx} the Penrose inequality in the second order of expansion is satisfied if
\begin{equation} \label{PIbound}
\int \displaylimits_M  |k'|^2_{g_S} \sqrt{1-\frac{2m}{\rho} + \frac{|\Lambda|}{3} \rho^2 }  d V_S \geq C  ||k'_{nn}||^2_{H^1 \left( M' \right)},
\end{equation}
where $C$ is a constant which depends on the local geometry of Schwarzschild-AdS near the boundary. Clearly this inequality is achieved for a large class of seed data, as $k'_{nn}$ contributes to but does not determine $|k'|_{g_S}$. Furthermore, we see that if a strict inequality is obtained in (\ref{PIbound}), then the error in the Penrose inequality becomes
\begin{equation}
\mathfrak{m} -\sqrt{\frac{A_h}{16 \pi}} \left(1+ \frac{ |\Lambda| }{12 \pi} A_h  \right) = c \epsilon^2 + O (\epsilon^3),
\end{equation}
for some constant $c>0$. It follows that the desired inequality holds for a large class of perturbed Schwarzschild-AdS initial data. Note that the outermost MOTS condition is not used here. However, we may view the special nature of the perturbations, and in particular the condition \eqref{PIbound}, as a replacement of this typical assumption.

\begin{theorem} \label{mth2}
Let $(M, \hat{g}=\phi^4 g_S, \hat{k}=\phi^{-2}k)$ be a 3-dimensional, vacuum, maximal, asymptotically hyperbolic, conformally perturbed Schwarzschild-AdS initial data set as in \eqref{fact_exp} of class $\left(\beta, \tau \right)$ with $\tau>3/2$, and having a MOTS inner boundary $\partial M$. If there exists a constant $C$ and domain $M'\supset \partial M$ such that
\begin{equation}
\int \displaylimits_M  |k|^2_{g_S} \sqrt{1-\frac{2m}{\rho} + \frac{|\Lambda|}{3} \rho^2 }  d V_S > C  ||k_{nn}||^2_{H^1 \left( M' \right)},
\end{equation}
then the Penrose inequality
\begin{equation}
\mathfrak{m} \geq  \sqrt{\frac{A_h}{16 \pi}} \left(1+ \frac{ |\Lambda| }{12 \pi} A_h  \right)
\end{equation}
holds for all $\epsilon$ sufficiently small, where $A_h$ is the area of the MOTS boundary.
\end{theorem}

\subsection{An outer-minimizing inner boundary}

In the context of the Penrose inequality, it may be desirable for the boundary MOTS to be either outermost or outerminimizing. Although the type of perturbations studied here do not require this hypothesis in order for the Penrose inequality to hold, it is useful to note that if $k_{nn}$ is positive on $\partial M$, then the class of initial data in Theorem \ref{mth2} has an outerminimizing boundary. This may be established by showing that the boundary satisfies the `shrink-wrap' obstacle problem. More precisely, in this case $M$ will admit a foliation by surfaces of positive mean curvature.

\begin{lemma}\label{aeoiroiqh}
Let $(M, \hat{g}=\phi^4 g_S, \hat{k}=\phi^{-2} k)$ be a 3-dimensional, vacuum, maximal, asymptotically hyperbolic, conformally perturbed Schwarzschild-AdS initial data set as in \eqref{fact_exp} of class $\left(\beta, \tau \right)$, $\tau>3/2$, with MOTS inner boundary $\partial M$. If $k_{nn} >0$ on $\partial M$, then each surface in the foliation by $\rho=\mathrm{const}$ spheres has positive mean curvature for sufficiently small $\epsilon$. In particular, the boundary is outerminimizing.
\end{lemma}

\begin{proof}
Observe that the mean curvature of the coordinate spheres $S_{\rho}\hookrightarrow(M,\hat{g})$ is given by
\begin{equation}
\hat{h}(\rho) =  4\phi^{-3} \partial_n \phi(\rho) + \phi^{-2}h(\rho),\quad\quad\quad
h(\rho)=\frac{2}{\rho} \sqrt{1- \frac{2m}{\rho} + \frac{|\Lambda|}{3} \rho^2},
\end{equation}
where $h(\rho)$ is the mean curvature of the same coordinate sphere sitting inside a constant time slice of the Schwarzschild-AdS spacetime. According to the expansion of (\ref{fact_exp}) we have
\begin{equation}
\hat{h}(\rho) = \frac{2}{\rho} \sqrt{1- \frac{2m}{\rho} + \frac{|\Lambda|}{3} \rho^2}  +  O (\epsilon).
\end{equation}
Moreover by (\ref{thetaplus}), since the inner boundary is a MOTS and the data is maximal it follows that $\hat{h} \left( \rho_h \right) = \hat{k}_{\hat{n} \hat{n}}(\rho_h)$. Using that $\hat{k}_{\hat{n} \hat{n}} = \phi^{-6} k_{nn}$, as well as the assumption $k_{nn} >0$ on $\partial M$, we see that this surface has positive mean curvature. Furthermore
\begin{equation}
\hat{h}(\rho) = 2 \sqrt{\frac{|\Lambda|}{3}} + O(\rho^{-2}) + O(\epsilon) \quad \mathrm{as} \quad \rho \rightarrow \infty,
\end{equation}
and
\begin{equation}
\hat{h}(\rho) =c\sqrt{\rho - \rho_h} + O \left( \left(\rho - \rho_h\right)^{3/2} \right) + O(\epsilon) \quad \mathrm{as} \quad \rho \rightarrow \rho_h,
\end{equation}
where $c>0$ is a constant. Hence,
\begin{equation}\label{hsmb1}
\hat{h}(\rho) >0 \quad \text{ if } \quad  \rho > \rho_h + c_1 \epsilon^2
\end{equation}
for $\epsilon$ sufficiently small, where $c_1$ is another constant. 

Next, observe that an expansion for $\hat{h}$ at the boundary can be computed with help from the second variation of area formula
\begin{equation}
\partial_{\hat{n}} \hat{h}(\rho_h) = - \hat{\mathrm{Ric}}(\hat{n}, \hat{n} ) - |\hat{A}|^2_{\hat{g}},
\end{equation}
where $\hat{\mathrm{Ric}}(\hat{n}, \hat{n} )$ is the Ricci curvature in the $\hat{n}$ direction and $\hat{A}$ is the boundary extrinsic curvature. Recall that two traces of the Gauss equations imply
\begin{equation}
\hat{\mathrm{Ric}}(\hat{n}, \hat{n} ) = \frac{1}{2} \hat{R} - \hat{K} + \frac{1}{2} \hat{h}^2 - \frac{1}{2} |\hat{A}|^2_{\hat{g}}     = - |\Lambda| + \frac{1}{2} |\hat{k}|^2_{\hat{g}} - \hat{K} + \frac{1}{2} \hat{h}^2 - \frac{1}{2} |\hat{A}|^2_{\hat{g}} ,
\end{equation}
where $\hat{K}$ is the Gaussian curvature of $S_{\rho_h}$ and the first vacuum constraint equation from (\ref{hamcon}) has been used in the second equality. Since on the boundary $\hat{h} = \hat{k}_{\hat{n} \hat{n}} =  \phi^{-6} k'_{nn} \epsilon$, and
\begin{equation}
\hat{A} = \frac{1}{2} h g_{S^2} + O(\epsilon)=O(\epsilon), \quad\quad\quad
\hat{K}=\rho_h^{-2}+O(\epsilon),
\end{equation}
we find that
\begin{equation}
\partial_{\hat{n}} \hat{h}(\rho_h) = |\Lambda|+\rho_h^{-2}+O(\epsilon).
\end{equation}
Let $\hat{t}(\rho)$ denote the function on $S_{\rho}$ which indicates the $\hat{g}$-distance to $\partial M$, then
\begin{equation}
c_2^{-1} \sqrt{\rho-\rho_h}\leq\hat{t}(\rho) = \bigintsss_{\rho_h}^{\rho} \frac{\phi^{2} d s}{\sqrt{1- \frac{2m}{s} + \frac{|\Lambda|}{3} s^2 }} \leq  c_2 \sqrt{\rho-\rho_h}
\end{equation}
for some constant $c_2>0$ when $\rho\in[\rho_h,\rho_h+1]$.
It follows that
\begin{align}
\begin{split}
\hat{h}(\rho)  = & \hat{h}(\rho_h) + \partial_{\hat{n}} \hat{h}(\rho_h) \hat{t}(\rho) + O\left( \hat{t}(\rho)^2 \right) \\
\geq & \phi^{-6}k_{nn}'(\rho_h) \epsilon+
c_2^{-1}(|\Lambda| + \rho_h^{-2} )\sqrt{\rho-\rho_h}\left(1-c_3\epsilon-c_4\sqrt{\rho-\rho_h}\right),
\end{split}
\end{align}
where the positive constants $c_3$ and $c_4$ are independent of $\epsilon$ and $\rho$. We then find that
\begin{equation}\label{hsmb2}
\hat{h}(\rho)\geq \phi^{-6}k_{nn}'(\rho_h) \epsilon>0 \quad\text{ if }\quad \rho<\rho_h +\min\{1,c_4^2 /2\},
\end{equation}
when $\epsilon$ is sufficiently small. Since $c_4$ does not depend on $\epsilon$, the desired result follows from the combination of (\ref{hsmb1}) and (\ref{hsmb2}).
\end{proof}

\subsection*{Acknowledgements}
JK would like to thank Jacek Tafel for introducing him to the topic of Penrose inequalities, and for encouragement along the way. This work was initiated during a visit to Stony Brook University, whose hospitality is greatly appreciated.

\end{document}